\documentclass{llncs}
\usepackage[T1]{fontenc}

\usepackage[title]{appendix}

\usepackage{amssymb}
\usepackage{xspace}
\usepackage{amsmath}
\usepackage{stmaryrd}
\usepackage{float}

\usepackage{graphicx}
\usepackage{tikz}
\usetikzlibrary{arrows,automata,decorations.pathreplacing,positioning,calc,arrows.meta,quotes,calligraphy,shapes.geometric}

\usepackage[shortlabels]{enumitem}

\newcommand{\uup}{{\sc uf{\scriptsize 1}}\xspace}
\newcommand{\ro}{{\sc ro}\xspace}
\newcommand{\iro}{{\sc iro}\xspace}
\newcommand{\idl}{{\sc idl}\xspace}
\newcommand{\rdl}{{\sc rdl}\xspace}

\newcommand{\irdluup}{{\sc uf{\scriptsize 1}$\cdot$irdl}\xspace}
\newcommand{\rdluup}{{\sc uf{\scriptsize 1}$\cdot$rdl}\xspace}

\newcommand{\lmix}{{\sc uf{\scriptsize 1}$\cdot$idl$\cdot$iro}\xspace}
\newcommand{\msoor}{{\sc uf{\scriptsize 1}$\cdot$ro}\xspace}
\newcommand{\msoiro}{{\sc uf{\scriptsize 1}$\cdot$iro}\xspace}

\newcommand{\ex}{\vspace{2pt}\textit{Example:}\ }

\newcommand{\pz}{P_{\textit{int}}\xspace}
\newcommand{\paxiom}{\textit{AXIOMS}_\textit{Supp}\xspace}

\newcommand{\startformula}{\textit{START}_\mathcal{M}\xspace}
\newcommand{\stepformula}{\textit{STEP}_\mathcal{M}\xspace}
\newcommand{\terminationformula}{\textit{END}_\mathcal{M}\xspace}
\newcommand{\haltformula}{\textit{HALT}_\mathcal{M}\xspace}

\newcommand{\rel}{\textit{Support}\xspace}

\newcommand{\onesucc}{\textit{Succ}_\textit{Supp}\xspace}
\newcommand{\onepred}{\textit{Pred}_\textit{Supp}\xspace}

\newcommand{\encodbeg}{\textit{EncodingBegins}\xspace}
\newcommand{\leftcvg}{\textit{NoPred}_\textit{Supp}\xspace}
\newcommand{\rsucc}{\textit{HasSucc}_\textit{Supp}\xspace}
\newcommand{\state}[1]{\textit{State}_{#1}\xspace}

\title{Decidability of Difference Logic over the Reals with Uninterpreted Unary Predicates}

\author{Bernard~Boigelot\orcidID{0009-0009-4721-3824} \and Pascal~Fontaine\orcidID{0000-0003-4700-6031} \and Baptiste~Vergain\orcidID{0009-0003-5545-4579}}

\institute{Montefiore Institute, B28\\
	Université de Liège, Belgium\\
	\email{\{Bernard.Boigelot,\,Pascal.Fontaine,\,BVergain\}@uliege.be}}

\begin{document}
\maketitle

\begin{abstract}
	First-order logic fragments mixing quantifiers, arithmetic, and uninterpreted predicates are often undecidable, as is, for instance, Presburger arithmetic extended with a single uninterpreted unary predicate.  In the SMT world, difference logic is a quite popular fragment of linear arithmetic which is less expressive than Presburger arithmetic.  Difference logic on integers with uninterpreted unary predicates is known to be decidable, even in the presence of quantifiers.  We here show that (quantified) difference logic on real numbers with a single uninterpreted unary predicate is undecidable, quite surprisingly. Moreover, we prove that difference logic on integers, together with order on reals, combined with uninterpreted unary predicates, remains decidable. 


\end{abstract}

\section{Introduction}
\label{sec:intro}

The success of satisfiability modulo theories (SMT) solvers in verification can be attributed to several things, but one of them is indisputably the omnipresence, in the combination of theories, of arithmetic reasoners.  As SMT solvers get stronger in quantified reasoning, it becomes more interesting to get a clear picture of decidability frontiers when arithmetic is used in a quantified SMT context.  Some pure arithmetic theories are already undecidable, even in their quantifier-free fragment, e.g., Peano arithmetic~\cite{Matiyasevich93}, i.e., a first-order theory of the natural numbers with addition and multiplication.  However, Presburger arithmetic, somehow the linear restriction of Peano arithmetic, is decidable even in the quantified case~\cite{haase2018survival}, but augmenting Presburger arithmetic with a single unary uninterpreted predicate already yields undecidability~\cite{downey1972undecidability,Halpern1,Speranski1}.  To obtain a decidable fragment mixing arithmetic and uninterpreted predicates, one must further restrict the expressiveness.

In the SMT world, difference logic used to be a popular fragment of arithmetic,
because of its low complexity in the quantifier-free case.  In this fragment,
arithmetic is limited to difference constraints of the form $x - y \bowtie c$
where $x$ and $y$ are variables, $c$ is an integer constant and $\bowtie$
belongs to $\{<, \leq, =, \geq, >\}$.  Difference constraints can, e.g., express conditions on
the distance between two variables, the atomic formula $x - y = 2$ stating that
the distance between the values of $x$ and $y$ must be exactly $2$.  Notice that
since difference constraints involve only two variables ($c$ is an integer
constant) those constraints are strictly less expressive than linear
constraints in Presburger arithmetic.  The decidability of the logic mixing
difference constraints and unary uninterpreted predicates, when interpreted over
$\mathbb{N}$ (or similarly $\mathbb{Z}$) reduces to the decidability of the
monadic second-order theory of one successor, usually referred to as S1S.  The
decidability of S1S has been established thanks to the concept of infinite-word
automaton~\cite{buchi1962decision}.

On the real domain, it is well known that the first-order theory of real-closed
fields, which is in a sense the real counterpart of Peano arithmetic, is
decidable~\cite{Tarski1} even in the presence of quantifiers.  Whereas this
might give the impression that decidability is more often obtained on the reals
than on the integers, we here prove that the logic mixing difference constraints
and unary uninterpreted predicates, when interpreted over $\mathbb{R}$, is
undecidable.

Further restricting the arithmetic language, and considering order on the real
domain only, it is known that the monadic second-order theory of order is
undecidable~\cite{gurevich1982monadic,shelah1975monadic}, but its universal
fragment is decidable~\cite{burgess1985decision}. In this work, we establish that
the fragment mixing unary uninterpreted predicates, difference constraints over
integer variables, and order constraints over real variables is decidable.

Section~\ref{sec:prelim} provides some prerequisites and the precise definition of the studied fragments.  In Section~\ref{sec:lmix}, we prove the decidability of the fragment mixing unary uninterpreted predicates, difference constraints over integer variables, and order constraints over real variables.  This was already the subject of a work-in-progress workshop paper~\cite{vergainSCSC2022}.  In Section~\ref{sec:undec}, we prove that the fragment of quantified difference constraints over real variables extended with a single unary uninterpreted predicate is undecidable.

\section{Preliminaries}
\label{sec:prelim}

We refer to e.g., \cite{enderton2001mathematical} for a general
introduction to first-order logic with equality, and assume that the
reader is familiar with the notions of signature, term, variable, and
formula.  We use the usual logical connectives ($\lor$, $\land$,
$\lnot$, $\Rightarrow$, $\Leftrightarrow$) and first-order quantification
$\exists x.\, \varphi$ and $\forall x.\, \varphi$, respectively equivalent to writing $\exists x\, (\varphi)$ and $\forall x\, (\varphi)$, i.e., the dot stands for an opening parenthesis that is closed at the end of the formula.
Variable symbols are denoted by $x$, $y$, $z$,\ \dots \ and are meant to be interpreted as real numbers.

Our signature contains the interpreted arithmetic symbols $0$, $1$, $+$, $-$, $<$,
$\leq$, $\geq$, $>$, $=$, and other constants in $\mathbb{N}$ that stand for terms $1 + 1 +
\cdots + 1$.  We furthermore use a monadic (i.e., unary) interpreted predicate $x \in\mathbb{Z}$ to
denote that $x$ has an integer value.  The signature also contains
uninterpreted predicate symbols $P$, $Q$,\ \ldots\ In the whole article, we only
consider unary predicate symbols. Indeed, including binary uninterpreted predicates without restriction on first-order quantification directly yields undecidability. Our language is the set of all well-formed
formulas, in the usual sense, built using symbols from the signature.  Further specific restrictions will be introduced later.

An interpretation specifies a domain (i.e.,\ a set of elements), assigns a value
in the domain to each free variable, and assigns relations of appropriate arity
on the domain to predicate symbols in the signature.  Throughout the article,
the interpretation domain is always $\mathbb{R}$.  The arithmetic symbols $0$,
$1$, $+$, $-$, $<$, $\leq$, $\geq$, $>$, $=$ are interpreted as expected on
$\mathbb{R}$, and $x\in\mathbb{Z}$ is true if and only if $x$ has an integer
value\footnote[2]{In the current context, this choice of notation for mixed
	integer-real arithmetic is simpler than using a multi-sorted logic.}.  An
interpretation assigns an arbitrary subset of the domain $\mathbb{R}$ to each
unary predicate.  By extension, an interpretation assigns a value in
$\mathbb{R}$ to every term, and a truth value to every formula. We denote the
interpretation $I$ of a variable $x$ by $I[x]$, and the interpretation of a
predicate $P$ by $I[P]$. A model of a formula is an interpretation that assigns
true to this formula.  A formula is satisfiable on a domain (here $\mathbb{R}$)
if it has a model on that domain.

\subsection{Difference Arithmetic with Unary Predicates}

We consider several fragments where the language is
restricted, in particular in the way that the arithmetic relations can be used.  A fragment is
decidable if there exists a procedure to check whether a given formula in
this fragment is satisfiable.

In the various fragments introduced below, all arithmetic atoms are either
\textit{order constraints} of the form $x\!\bowtie\!y$, or \textit{difference
  constraints} of the form $x - y \bowtie c$, where $x$ and $y$ are variables,
$c$ is a constant in $\mathbb{Z}$, and $\bowtie \ \in \{<, \leq, =, \geq, >\}$.
As a reminder, the language of our formulas only contains \textit{unary	predicates}. 
The only atoms besides the arithmetic ones are of the form
$P(x)$ where $P$ is an uninterpreted predicate symbol and $x$ is a variable, and
$x \in \mathbb{Z}$ where $x$ is a variable. 
Note that the addition of constraints of the form $x \bowtie c$, where $x$ is a variable and $c$ is an integer constant, to fragments that already admit difference constraints does not increase their expressive power: a constraint $x \bowtie c$ is equivalent to the difference constraint $x-v_0 \bowtie c$, where $v_0$ is a particular variable intended to be interpreted as zero.
It is worth mentioning that the models of the formulas in our fragments are invariant by an integer shift, i.e., mapping every real $x$ to the real $x + j$, where $j$ is an integer.
Conjunctions of order constraints will be merged to improve readability, i.e., we will often write $x<y<z$ rather than $x<y \land y<z$. 
Finally, we use the shorthand $P(x+c)$ instead of $\exists y.\, y - x = c \land P(y)$, where $x$ is a free variable and $c\in \mathbb{Z}$.

We now introduce our fragments of interest. Their names are inspired from the SMT-LIB nomenclature, where acronyms stand for the theories that appear in the combinations:
\begin{itemize}
	\item \uup: the theory of uninterpreted functions, with the restriction that uninterpreted symbols may only correspond to monadic predicates;
	\item \ro: the theory of order on the reals only;
	\item \iro: the theory of order on the reals and integers;
	\item \idl: difference logic on the integers;
	\item \rdl: difference logic on the reals.
\end{itemize}

\paragraph{{\bf\msoor}.}
The fragment \msoor is the fragment with unary uninterpreted predicates and
order constraints between variables interpreted over $\mathbb{R}$.
Difference logic constraints and atoms of the form $x\in\mathbb{Z}$ are not
allowed.

\ex The formula	$\forall x\,\exists\,y,z\,.\, y<x<z \land \forall t\,.\, (y<t<z \land P(t) ) \Rightarrow t= x$ describes a predicate $P$ that is true only on isolated real numbers.

\paragraph{{\bf\msoiro}.}
The fragment \msoiro is the extension of \msoor where atoms of the form
$x\in\mathbb{Z}$ are allowed.  This fragment can express order relations between
real and integer variables.

\ex The formula $\forall x, y.\, (x<y \land x \in \mathbb{Z} \land y \in \mathbb{Z}) \Rightarrow \exists v.\, x<v<y \land P(v)$ describes a predicate $P$ that is true for at least one value located between any two integers.

\paragraph{{\bf\lmix}.}
The fragment \lmix is an extension of the fragment \msoiro (and therefore of
\msoor). It is also interpreted over $\mathbb{R}$.  Order constraints between
variables and atoms of the form $x\in\mathbb{Z}$ are allowed. Additionally,
difference logic constraints are allowed, but they can only involve \textit{integer-guarded} variables.

In order to enforce this integer-guard restriction on difference logic
constraints, \lmix formulas must be \textit{well-guarded}, i.e.,
difference logic constraints can only appear in the two following contexts:
\begin{itemize}
	\item $x \in \mathbb{Z} \land y \in \mathbb{Z} \land x - y \bowtie c$, 
	\item $(x \in \mathbb{Z} \land y \in \mathbb{Z}) \Rightarrow x - y \bowtie c$,
\end{itemize} 
where $x$ and $y$ are variables, $c \in \mathbb{Z}$ is a constant, and $\bowtie \ \in \{<, \leq, =, \geq, >\}$.

\ex The following formula describes a predicate that is either true on all odd numbers and false on all even numbers, or the opposite, as well as true on all non-integer numbers:\\
$\mbox{\hspace{20pt}} \big[ \forall x,y.\, \big(x \in \mathbb{Z} \land y \in \mathbb{Z} \land y-x=2 \big) \Rightarrow \big( P(x) \Leftrightarrow P(y)\big) \big]\\
\mbox{\hspace{20pt}} \land \big[\exists x,y.\, x \in \mathbb{Z} \land y \in \mathbb{Z} \land P(x) \land \lnot P(y)\big] \land \big[\forall z.\, \lnot(z \in \mathbb{Z}) \Rightarrow P(z)\big]$

\paragraph{{\bf\rdluup}.}
The fragment \rdluup is the fragment interpreted over $\mathbb{R}$, where order
constraints, difference logic constraints and unary predicate atoms are allowed
without any restriction. The use of atoms of the form $x\in\mathbb{Z}$ is
forbidden. Since order constraints are a special case of difference logic constraints, the name of the fragment only refers to \rdl and not \ro.

\ex The formula $\forall x\, \exists y.\, 0 < y-x < 3 \land P(y)$ describes a predicate $P$ such that any subinterval of $\mathbb{R}$ of length greater or equal to $3$ contains a value for which $P$ is true.

\vspace{3pt}\textit{Note:}\ It might appear to the reader that a missing logic in this nomenclature is \irdluup, with difference logic constraints on both real and integer variables.  We will later show that \rdluup is already undecidable, so it makes little sense to introduce any extension of it.

\section{Decidability of \lmix}
\label{sec:lmix}

The fragment \msoor is actually a restriction of the universal fragment of the monadic second-order theory of the real order $\mathbb{R}$, i.e., \msoor augmented with universal quantification of predicate variables. It has been established in~\cite{burgess1985decision} that the universal fragment of the monadic second-order theory of the real order $\mathbb{R}$ is decidable, which trivially implies the decidability of \msoor. We show here that its extension \lmix (and therefore \msoiro) is also decidable, by a reduction to \msoor.

\begin{theorem}
	\lmix and \msoiro are decidable.
\end{theorem}

Note that the decidability of \msoiro is a direct consequence of the decidability of \lmix, since \lmix is an extension of \msoiro.
The remaining of this section is thus dedicated to proving that \lmix is decidable.

\subsection{Recognizing Integer Values} 
\label{subsec:recog}

We first show how to define in \msoor a predicate $\pz$ over $\mathbb{R}$ that is \mbox{$<$-isomorphic} to $\mathbb{Z}$, i.e., such that there exists a bijection between the sets described by $\pz$ and $\mathbb{Z}$ that preserves the order relation over their elements.  Integer guards in \lmix will later be translated using $\pz$.  Intuitively, an integer-guarded variable in a \lmix formula will correspond to a variable taking its value in the set described by $\pz$ in the translated \msoor formula.

We axiomatize $\pz$ in \msoor as follows:\\
\mbox{\hspace{16pt}} $\bullet$ Every element of $\pz$ is isolated:\\
\mbox{\hspace{25pt}} $\forall x\, \exists\, y,z.\ y < x < z \land \forall t.\ [y < t < z \land \pz(t)] \Rightarrow t = x $.\\
\mbox{\hspace{16pt}} $\bullet$ Every point in $\mathbb{R}$ has a unique successor in $\pz$:\\
\mbox{\hspace{25pt}} $\forall x\, \exists\, y.\ x < y \land \pz(y) \land \forall t.\ x < t < y \Rightarrow \neg \pz(t)$.\\
\mbox{\hspace{16pt}} $\bullet$ Similarly, every point in $\mathbb{R}$ has a unique predecessor in $\pz$:\\
\mbox{\hspace{25pt}} $\forall x\, \exists\, y.\ y < x \land \pz(y) \land \forall t.\ y < t < x \Rightarrow \neg \pz(t)$.


The set of all integers is a model for $\pz$, therefore the above axiomatization is consistent.
The set of elements satisfying $\pz$ is necessarily infinite and does not admit a maximal or a minimal element. This is a direct consequence of the successor and predecessor axioms.
More interestingly, this set is also necessarily countable. Indeed, since each point is isolated, there exists an application that maps the elements satisfying $\pz$ to disjoint open intervals. Any set of disjoint intervals in $\mathbb{R}$ with non-zero length is necessarily countable~\cite{sierpinski1958cardinal}, since each of them contains a rational value that does not belong to the others. \\
It is now possible to define a successor relation on the real numbers satisfying $\pz$ with the formula  $\textit{Succ}(x,y) \! = \! \pz(x) \land \pz(y) \land y \! < \! x \land \forall z.\, y \! < \! z \! < \! x \Rightarrow \lnot \pz(z)$, i.e., $x$ is the successor of $y$, or equivalently, $y$ is the predecessor of $x$. \\
The axiomatization of $\pz$ is, in fact, precise enough to have the following lemma.

\begin{lemma}
	\label{lemma:PINTisoZ}
	For any model $M$ of $\pz$, the set $M[\pz]$ is $<$-isomorphic to $\mathbb{Z}$.
\end{lemma}

For convenience in the proof, we define $0_{\textit{int}}$ as an arbitrary existentially quantified value that belongs to the set described by $\pz$.

\begin{proof}	
	Given a model $M$ of the axiomatization of $\pz$, we need to define a bijection between the set $M[\pz]$ and $\mathbb{Z}$ that preserves order. 
	
	Let us define an application $f$ from $M[\pz]$ to $\mathbb{Z}$.  We set $f(0_{\textit{int}}) = 0$, and then define recursively: 
	\begin{itemize}
		\item $f(y) = f(x) + 1$ for each $x,y \in M[\pz]$ such that $y > 0_{\textit{int}}$ and $\textit{Succ}(y,x)$,
		\item $f(y) = f(x) - 1$ for each $x,y \in M[\pz]$ such that $y < 0_{\textit{int}}$ and $\textit{Succ}(x,y)$.
	\end{itemize}
	
	Thanks to the fact that every element of $M[\pz]$ has a unique predecessor and successor, it
	follows that $f$ ranges over the whole set $\mathbb{Z}$, proving that $f$ is surjective. Since it is clear that
	$f$ preserves order, it follows that $f$ is strictly increasing, and therefore injective.  It remains to show that $f$ is well defined for every
	element in $M[\pz]$.
	
	If there exists some element $y \in M[\pz]$ for which $f$ is not
	defined, it means that $f$ is not well-defined, in the sense that there exists either an element
	$y > 0_{\textit{int}}$ such that the interval $[0_{\textit{int}},y]$ contains an infinite number of elements
	satisfying $\pz$, or there exists an element $y < 0_{\textit{int}}$ such that the interval
	$[y,0_{\textit{int}}]$ contains an infinite number of elements satisfying $\pz$.  Since both cases are
	symmetric, we only address the former.  There must exist a strictly increasing
	infinite series of elements in $M[\pz]$ bounded by $y$.  Let us
	consider its limit $z\in \mathbb{R}$.  Because there must exist an element of
	$M[\pz]$ smaller than $z$ and arbitrarily close to $z$, it follows that $z$ cannot have a predecessor, which contradicts an axiom.  Therefore $f$ is well-defined, and
	every element of $M[\pz]$ is associated to an integer number. The application $f$ is therefore a bijection. \qed
\end{proof}

\subsection{Translating Formulas}

We are now able to describe the satisfiability-preserving translation of
formulas from \lmix to \msoor.  Consider a \lmix formula $\varphi$.  Without
loss of generality, we assume that $\pz$ does not appear in $\varphi$.  The
translation of $\varphi$ is defined as
$$\textit{AXIOMS}_{int}(\pz) \land \llbracket\varphi\rrbracket$$ 
where $\textit{AXIOMS}_{int}(\pz)$ is the conjunction of the axioms of $\pz$, and $\llbracket \cdot \rrbracket$ is a translation operator. This translation operator $\llbracket \cdot \rrbracket$ distributes over all Boolean operators and quantifiers, and corresponds to the identity transformation for most considered atoms, except in the following cases: 
\begin{itemize}
\item $\llbracket x \in \mathbb{Z} \rrbracket = \pz(x)$;
\item $\llbracket x - y \bowtie c \rrbracket = \exists z_0,\dots z_c\,.\  (y = z_0) \land (x \bowtie z_c) \land \bigwedge_{0 \leq i < c} \textit{Succ}(z_{i+1}, z_i)$,\\ for $c \in \mathbb{N}$ and $\bowtie \ \in \{<, \leq, =, \geq, >\}$. We assume that $z_0, \dots z_c$ are fresh variables w.r.t.\ $x$ and $y$.
\end{itemize} 
\ex $\llbracket x - y \leq 2 \rrbracket = \exists z_0, z_1, z_2.\ y = z_0 \land \textit{Succ}(z_1, z_0) \land \textit{Succ}(z_2, z_1) \land x \leq z_2$. \\
Notice that we only deal with the case $c \! \in \! \mathbb{N}$ since every atom of the form $x - y \bowtie c$ with $c \in \mathbb{Z} \setminus \mathbb{N}$ and $\bowtie \ \in \{<, \leq, =, \geq, >\}$ can be rewritten as $y - x \bowtie^\prime -c$ with the following correspondences: $(\bowtie, \bowtie') \in \{{(=, =)}, {(<, >)}, {(>, <)}, {(\geq, \leq)}, {(\leq, \geq)}\}$.

\subsection{Establishing Equisatisfiability}

Given a \lmix formula $\varphi$, the translation that we have introduced generates a corresponding \msoor formula $\psi$. To establish that they are equisatisfiable, we need to prove that if $\varphi$ admits a model, then $\psi$ also admits one, and reciprocally.

\begin{lemma}
	Given a \lmix formula $\varphi$, consider its translation into \msoor $\psi = \textit{AXIOMS}_{int}(\pz) \land \llbracket\varphi\rrbracket$. The formulas $\varphi$ and $\psi$ are equisatisfiable.
\end{lemma}

\begin{proof}
If $\varphi$ is satisfiable, let $M$ be one of its models. Then, since $\psi$ shares the same free variables and predicates than $\varphi$ with the only addition of $\pz$, we can directly construct a model $M'$ of $\psi$ that is similar to $M$ for the shared variables and predicates, and that interprets $\pz$ so that $\pz(x)$ holds whenever $x \in \mathbb{Z}$. This is always possible since the only constraints on $\pz$ generated by the construction of $\psi$ are the axioms stated above.

\vspace{2pt}
If $\psi$ is satisfiable, then there exists a model $M$ of $\psi$. Let us construct a model $M'$ of $\varphi$. Let $0_\textit{int} \in \mathbb{R}$ be an arbitrary element of $M[\pz]$. We define an automorphism $g$ of $\mathbb{R}$, such that $g(0_\textit{int}) = 0$, and recursively $g(y) = g(x) + 1$ for $x,y \in M[\pz]$, $y>0_\textit{int}$ and $\textit{Succ}(y,x)$, and $g(y) = g(x) - 1$ for $x,y \in M[\pz]$, $y<0_\textit{int}$ and $\textit{Succ}(x,y)$. The automorphism $g$ maps each open interval between the $k$-th and \mbox{$(k\!+\!1)$-th} successors (resp.\ predecessors) of $0_\textit{int}$ in $M[\pz]$, onto the open interval $(k, k+1)$ (resp.\ $(-(k\!+\!1), -k)$) while preserving order.

$M'$ is defined by $M'[x] = g(M[x])$ for each free variable $x$ of the formula $\varphi$, and $M'[P] = \{g(x) \, | \, x \in M[P]\}$ for each uninterpreted predicate $P$ of $\varphi$. No unary predicate atom can be violated by $M'$ by definition. Furthermore, no order constraint can be violated by $M'$ either since $g$ preserves order. Regarding the difference logic constraints, the intermediate variables $z_i$ introduced in the translation are necessarily mapped to values in $M[\pz]$ since the \textit{Succ} relation enforces this property. Hence for each such variable, we have $g(M[z_i]) \in \mathbb{Z}$. Intuitively, this ensures that in $M'$ the difference between the values taken by the integer variables is consistent with the difference logic constraints.
It follows that $M'$ is a model of $\varphi$. 
\qed
\end{proof}

\section{Undecidability of \rdluup}
\label{sec:undec}

The result presented in the previous section establishes a lower bound for the decidability of our family of fragments.  A natural follow-up problem is to establish a corresponding upper bound, i.e., to find an extension of this logic that yields undecidability.  We show here that, when combined with uninterpreted unary predicates, as soon as difference logic constraints on reals are allowed, the logic becomes undecidable.  

We actually show a stronger result which is that a single unary predicate symbol is enough to yield undecidability. More precisely, we establish the undecidability of the restriction of \rdluup where only one predicate symbol is allowed, by reducing the halting problem of a Turing machine to the satisfiability problem over this restriction of \rdluup.

\begin{theorem}
	\label{thm:undec}
	Satisfiability is undecidable for \rdluup with a single predicate.
\end{theorem}

\begin{corollary}
	Satisfiability is undecidable for \rdluup.
\end{corollary}

The remaining of this section is dedicated to proving Theorem~\ref{thm:undec}. We consider w.l.o.g.\ Turing machines defined over an alphabet with only two symbols and no explicit blank symbol~\cite{shannon1956universal}. This choice leads to a simpler proof. 

\subsection{Definitions}

The proof is by reduction from the halting problem for a Turing machine with a single bi-infinite tape, 
starting from a blank tape (i.e., a tape filled with the symbol $0$).  Consider a Turing machine 
${\cal M} = (Q, \Sigma, q_I, q_F, \Delta)$, where
\begin{itemize}
	\item
	$Q$ is a finite nonempty set of states,
	\item
	$\Sigma$ is the alphabet $\{ 0, 1 \}$,
	\item
	$q_I \in Q$ is the initial state,
	\item
	$q_F \in Q$ is the halting state,
	\item
	$\Delta \subseteq \{(Q \setminus \{ q_F\}) \times \Sigma
	\times Q \times \Sigma \times \{L, R\} \}$
	is the transition relation, assumed to be total over its first two components, i.e., for any pair $(q, \alpha) \in (Q \setminus \{ q_F\}) \times \Sigma$, there exists a tuple $(q, \alpha, q', \alpha', \lambda) \in \Delta$.
\end{itemize}
A \textit{configuration} $C$ of such a Turing machine is a triplet containing the current state $q$, the content of the tape $t \in \{0,1\}^\mathbb{Z}$ and the position of the head $h \in \mathbb{Z}$.
Since the machine starts from a blank tape, the initial configuration is $C_0 = (q_I, 0^\mathbb{Z}, 0)$.

A \textit{run} $\rho$ of length $n \in \mathbb{N}$ (resp.\ $n = + \infty$) of such a Turing machine is a finite (resp.\ infinite) sequence of configurations $(C_i)_{i \in [0,n]}$ (resp.\ $(C_i)_{i \in \mathbb{N}}$), such that for any two consecutive configurations $C_i = (q_i, t_i, h_i)$ and $C_{i+1} = (q_{i+1}, t_{i+1}, h_{i+1})$ there exists a transition $(q, \alpha, q', \alpha', \lambda) \in \Delta$ such that:

\begin{itemize}
	\item 
	$q=q_i$ and $q' = q_{i+1}$,
	\item 
	$t_i[h_i] = \alpha$, i.e., the tape cell at position $h_i$ contains the symbol $\alpha$,
	\item 
	$t_{i+1}[h_i] = \alpha'$,
	\item 
	$t_{i+1}[k] = t_i[k]$, for every $k\in \mathbb{Z}$,  $k \neq h_i$,
	\item 
	$h_{i+1} = h_i + 1$ if $\lambda = R$, and $h_{i+1} = h_i - 1$ if $\lambda = L$.
\end{itemize}
A \textit{halting run} is a finite run such that the state of its last configuration is the halting state $q_F$.

\subsection{Encoding Runs}
Our goal is to encode a run of a Turing machine (as described before), i.e., encode the state, the tape content, and the position of the head for each configuration of such a run. 
Starting from the initial configuration, we must also ensure the coherence of the run w.r.t.\ the Turing machine transition relation, by connecting every two consecutive configurations.
Our idea is to define an infinite sequence of intervals on the real line, such that each interval contains the encoding of its corresponding configuration (i.e., the first interval will contain the first configuration of the run, and so on). Difference constraints can then be used to connect consecutive configurations. 

Let $N = \lceil\log_2 (|Q|) \rceil$. Each state $q \in Q$ of $\mathcal{M}$ can therefore be uniquely encoded with $N$ Boolean values $b^q_1, \dots b^q_N$.
We want to encode consecutive configurations of the Turing machine using a single predicate $P$ over $\mathbb{R}$. In order to do so, we first need to describe a subset of $\mathbb{R}$ that will act as a grid supporting the encoding of the state, the tape content, and the head position of the current configuration.

We use the concept of linear ordering~\cite{rosenstein1982linear} to describe the shape of the grid. A \emph{linear ordering} $J$ is a totally ordered set, i.e., a set equipped with a binary relation $<$ which is irreflexive (for all $j$ in $J$,  $j\not<j$), asymmetric (for all $j,k$ in $J$, if $j<k$, then $k\not<j$), transitive (for all $i,j,k$ in $J$, if $i<j$ and $j<k$, then $i<k$), and complete (for all $j,k \in J$, either $j = k$, $j<k$, or $k<j$). The \emph{order type} of a linear ordering $J$ is the class of all linear orderings $<$-isomorphic to $J$.
The order types of a singleton, the set composed of the $N$ first natural numbers, $\mathbb{N}$, and $\mathbb{Z}$ are respectively denoted by $1$, $N$, $\omega$, and $\zeta$. 
The concatenation of two linear orderings $J$ and $K$ (where their associated order relations are respectively $<_J$ and $<_K$) is denoted by $J + K$. It corresponds to the linear ordering composed of the set of pairs $\{(j, 1) \,|\, j \in J\} \cup \{(k, 2)\,|\, k \in K\}$, and equipped with the order relation $<$, defined by $(j_1, 1) < (j_2, 1)$ if $j_1 <_J j_2$, $(k_1, 2) < (k_2, 2)$ if $k_1 <_K k_2$, and $(j, 1)<(k, 2)$ for every $j\in J$ and $k\in K$.
More generally, given two linear orderings $J$ and $K$, the linear ordering $(J)^K$ is the set of pairs $(j, k)$ with $j \in J$ and $k \in K$, with the order relation $<$ such that $(j_1, k_1) < (j_2, k_2)$ if either $k_1 <_K k_2$, or $k_1 =_K k_2$ and $j_1 <_J j_2$. These operators are naturally extended on order types. For instance, the order type $(\omega)^\omega$ is the class of all linear orderings $<$-isomorphic to $\mathbb{N}^2$.

The grid we consider is a linear ordering that is a subset of $\mathbb{R}$, of order type $\big(N + \zeta + 1 + \zeta \big)^{\omega}$. An ordering of order type $N + \zeta + 1 + \zeta$ within the interval $[0,3)$ is depicted in Figure~\ref{fig:DepictLinOrder}. Each dot corresponds to a natural number and each vertical line corresponds to an element of the linear ordering. The first $N$ points will support the encoding of a state. The first subordering that is \mbox{$<$-isomorphic} to $\mathbb{Z}$ (i.e., of order type $\zeta$) will be used to encode the position of the head, while the second one will support the encoding of the tape content. The whole grid is composed of an infinite repetition of the subordering $N + \zeta + 1 + \zeta$ (i.e., it is repeated on the intervals $[3k,3k+3)$ for all $k \in \mathbb{N}$), hence the $\omega$ exponent.
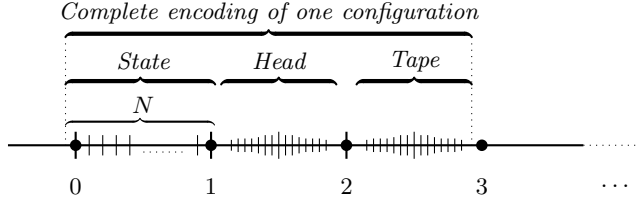
\begin{figure}[h]
	\centering
		\begin{tikzpicture}[scale=.9,node distance=2cm]
			\path[use as bounding box] (-1, 4.6) rectangle (8.5, 7);
			\node[] at (0,4.4) {$0$};
			\node[] at (2,4.4) {$1$};
			\node[] at (4,4.4) {$2$};
			\node[] at (6,4.4) {$3$};
			\node[] at (8,4.4) {$\dots$};
			
			\draw[thick] (0,5.2) -- (0,4.8);
			\draw[thick] (2,5.2) -- (2,4.8);
			\draw[thick] (4,5.2) -- (4,4.8);
	
			\node[circle, fill, scale=.5] at (0,5) {};
			\node[circle, fill, scale=.5] at (2,5) {};
			\node[circle, fill, scale=.5] at (4,5) {};
			\node[circle, fill, scale=.5] at (6,5) {};
			
			\draw (3,5.2) -- (3,4.8);
			\draw (3.1,5.15) -- (3.1,4.85);
			\draw (2.9,5.15) -- (2.9,4.85);
			\draw (3.2,5.15) -- (3.2,4.85);
			\draw (2.8,5.15) -- (2.8,4.85);
			\draw (3.3,5.12) -- (3.3,4.88);
			\draw (2.7,5.12) -- (2.7,4.88);
			\draw (3.4,5.1) -- (3.4,4.9);
			\draw (2.6,5.1) -- (2.6,4.9);
			\draw (3.5,5.1) -- (3.5,4.9);
			\draw (2.5,5.1) -- (2.5,4.9);
			\draw (3.6,5.1) -- (3.6,4.9);
			\draw (2.4,5.1) -- (2.4,4.9);
			\draw (3.7,5.08) -- (3.7,4.92);
			\draw (2.3,5.08) -- (2.3,4.92);
			
			\draw (0.2,5.15) -- (0.2,4.85);
			\draw (0.4,5.15) -- (0.4,4.85);
			\draw (0.6,5.15) -- (0.6,4.85);
			\draw (0.8,5.15) -- (0.8,4.85);
			\draw[dotted] (1, 4.9) -- (1.6, 4.9);
			\draw (1.8,5.15) -- (1.8,4.85);
					
			\draw[thick] [decorate, decoration = {calligraphic brace}] (-0.15,5.3) -- (2.05,5.3);
			\node[] at (1,5.6) {$N$};
					
			\draw (5,5.2) -- (5,4.8);
			\draw (5.1,5.15) -- (5.1,4.85);
			\draw (4.9,5.15) -- (4.9,4.85);
			\draw (5.2,5.15) -- (5.2,4.85);
			\draw (4.8,5.15) -- (4.8,4.85);
			\draw (5.3,5.12) -- (5.3,4.88);
			\draw (4.7,5.12) -- (4.7,4.88);
			\draw (5.4,5.1) -- (5.4,4.9);
			\draw (4.6,5.1) -- (4.6,4.9);
			\draw (5.5,5.1) -- (5.5,4.9);
			\draw (4.5,5.1) -- (4.5,4.9);
			\draw (5.6,5.1) -- (5.6,4.9);
			\draw (4.4,5.1) -- (4.4,4.9);
			\draw (5.7,5.08) -- (5.7,4.92);
			\draw (4.3,5.08) -- (4.3,4.92);			
			
			\draw[thick] (-1, 5) -- (7.5,5);
			\draw[dotted] (7.5,5) -- (8.5,5);
			
			\draw[ultra thick] [decorate, decoration = {calligraphic brace}] (-0.15,5.9) -- (2.06,5.9);
			\node[] at (1,6.25) {$\textit{State}$};
			
			\draw[ultra thick] [decorate, decoration = {calligraphic brace}] (2.15,5.9) -- (3.85,5.9);
			\node[] at (3,6.25) {$\textit{Head}$};
			
			\draw[ultra thick] [decorate, decoration = {calligraphic brace}] (4.15,5.9) -- (5.85,5.9);
			\node[] at (5,6.25) {$\textit{Tape}$};
			
			\draw[ultra thick] [decorate, decoration = {calligraphic brace}] (-0.15,6.6) -- (5.85,6.6);
			\node[] at (2.85,6.95) {$\textit{Complete encoding of one configuration}$};			
			
			\draw[dotted] (-0.15, 6.6) -- (-0.15, 5);
			\draw[dotted] (5.85, 6.6) -- (5.85, 5);
		\end{tikzpicture}
		\caption{A visual representation of a linear ordering of order type $N + \zeta + 1 + \zeta$.}
		\label{fig:DepictLinOrder}
\end{figure}

\subsection{Defining the Support of the Encoding}
\label{subsec:support}

Let us first define concretely the support of the encoding of the Turing machine configurations.
The difficulty lies in describing the grid using a single predicate $P$, without meddling with the actual encoding of the configurations afterwards. Our solution is to characterize the points that belong to the grid by enforcing that such a point is surrounded by an open interval where $P$ is uniformly \textit{true} on the left, and by an open interval where $P$ is uniformly \textit{false} on the right, such as depicted in Figure~\ref{fig:SingleRelPoint}. We do not specify yet how $P$ behaves on $x$, as this is how the configurations will actually be encoded later.
\begin{figure}[b]
	\centering
		\begin{tikzpicture}
			\path[use as bounding box] (3, 5) rectangle (7, 5.3);
			\filldraw[] (3,4.85) rectangle ++(2,0.3);
			\draw[] (5,4.85) rectangle ++(2,0.3);
			\node at (4, 5.4) {$P$};
			\node at (6, 5.4) {$\lnot P$};
			\draw (5,5.25) -- (5,4.75);
			\node at (5,5.4) {$x$};
		\end{tikzpicture}
	\caption{The real number $x$ belongs to the grid, since it is surrounded by a \textit{true} (black) open interval on the left, and a \textit{false} (white) open interval on the right.}
	\label{fig:SingleRelPoint}
\end{figure}
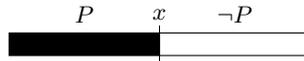

Such a characterization is easy to express in our restriction of \rdluup:
$$\rel(x) \! = \! (\exists y.\, y \! <\! x \land \forall z.\, y \! < \! z \! < \! x \Rightarrow P(z)) \land (\exists y.\, x \! < \! y \land \forall z.\, x \! < \! z \! < \! y \Rightarrow \lnot P(z))$$

Let us now partially axiomatize the predicate $P$ such that the set of \textit{supporting} points constitutes a linear ordering of order type $\big(N + \zeta + 1 + \zeta \big)^{\omega}$:
\begin{enumerate}[(a)]
	\item Let $\mathbf{0}$ be a variable and $\mathbf{1}, \mathbf{2}$ and $\mathbf{3}$ be respectively the $+1$-successor of $\mathbf{0}$, $\mathbf{1}$ and $\mathbf{2}$:\\
	  $\textit{Axiom}_1 = (\mathbf{1} = \mathbf{0} + 1) \land (\mathbf{2} = \mathbf{1} + 1) \land (\mathbf{3} = \mathbf{2} + 1)$\\
           These free variables are implicitly existentially quantified in the final formula.
          
	Notice that the variable $\mathbf{0}$ can be interpreted as any real value, which only acts as a landmark for the beginning of the grid.
	
	\item $\mathbf{0}$, $\mathbf{1}$ and $\mathbf{2}$ are \emph{supporting} points:\\
	$\textit{Axiom}_2 = \rel(\mathbf{0}) \land \rel(\mathbf{1}) \land \rel(\mathbf{2})$
	
	\item $P$ is uniformly true before $\mathbf{0}$, i.e., there are no supporting points before $\mathbf{0}$:\\
	$\textit{Axiom}_3 = \forall x.\, x<\mathbf{0} \Rightarrow P(x)$
	
	\item There are exactly $N-2$ \textit{supporting} points within the interval $(\mathbf{0},\mathbf{1})$:\\
	$\textit{Axiom}_4 = \exists x_1, x_2, \dots x_{N}.\, x_1 = \mathbf{0} \land x_{N} = \mathbf{1} \\ 
	\mbox{\hspace{45pt}} \land \bigwedge_{1\leq i<N} \big(\mathbf{0} \leq x_i< \mathbf{1} \land \onesucc(x_{i+1}, x_i)\big)$
	
	where $\onesucc(x,y)$ is a formula that states that $x$ is the first \textit{supporting} real value that is strictly greater than $y$, i.e., $x$ is the successor of $y$ on the grid. It is defined as follows:\\ 
	$\onesucc(x,y) \! = \! y \! < \! x \land \rel(x) \land \rel(y) \land \forall z.\, y \! < \! z \! < \! x \Rightarrow \lnot \rel(z)$\\
	We also define an analogous formula to express that $x$ is the predecessor of $y$: $\onepred(x, y) = \onesucc(y, x)$.
	
	\item The set of \textit{supporting} points within $(\mathbf{1},\mathbf{2})$ is $<$-isomorphic to $\mathbb{Z}$. This is done similarly to the axiomatization of $\pz$ (cf. Section \ref{subsec:recog}). But because $\mathbf{1}$ (resp.\ $\mathbf{2}$) is a \textit{supporting} point, there must exist a uniformly false (resp.\ true) interval of $P$ at its right (resp.\ left) where no other \textit{supporting} points can appear. All the \textit{supporting} points will therefore be constrained to appear within a smaller interval $(b_1, b_2)$ with \mbox{$\mathbf{1}<b_1<b_2<\mathbf{2}$}, as illustrated in Figure~\ref{fig:1predicate}.
	\begin{align}
		\textit{Axiom}_5 = \,& [\exists b_1, b_2.\, \mathbf{1}<b_1<b_2<\mathbf{2}] \label{line:case1}\\
		\land \ &[\forall x.\, (b_1<x<b_2) \Rightarrow \exists y.\, x<y<b_2 \land \rel(y) \nonumber \\
		& \qquad \qquad \qquad \qquad \qquad \qquad \land \forall z.\, x<z<y \Rightarrow \lnot \rel(z)] \label{line:case3}\\
		\land \ &[\forall x.\, (b_1<x<b_2) \Rightarrow \exists y.\, b_1<y<x \land \rel(y) \nonumber \\
		& \qquad \qquad \qquad \qquad \qquad \qquad \land \forall z.\, y<z<x \Rightarrow \lnot \rel(z)] \label{line:case4}\\
		\land \ &[\forall x.\, (1<x<2 \land \rel(x)) \Rightarrow b_1<x<b_2] \label{line:case5}
	\end{align}
	This axiom can be broken down into these elementary pieces:
	\begin{enumerate}
		\item[(\ref{line:case1})] there exists an open interval $(b_1, b_2)$ such that $\mathbf{1}<b_1<b_2<\mathbf{2}$,
		
		
		\item[(\ref{line:case3})] each real value in $(b_1, b_2)$ has a \textit{supporting} successor,
		
		\item[(\ref{line:case4})] each real value in $(b_1,b_2)$ has a \textit{supporting} predecessor,
		
		\item[(\ref{line:case5})] there are no \textit{supporting} points within $(\mathbf{1}, b_1)$, nor within $(b_2, \mathbf{2})$.
	\end{enumerate}

	\item The pattern of \textit{supporting} points within $(\mathbf{1},\mathbf{2})$ is repeated onto the interval $(\mathbf{2},\mathbf{3})$ with an exact offset of $1$: \\
	$\textit{Axiom}_6 = \forall x.\, \mathbf{1}<x<\mathbf{2} \Rightarrow (\rel(x) \Leftrightarrow \rel(x+1))$ 
	
	\item The pattern of \textit{supporting} points within $[\mathbf{0},\mathbf{3})$ is repeated onto every interval $[\mathbf{3k}, \mathbf{3k+3})$ for $k \in \mathbb{N}$: \\
	$\textit{Axiom}_7 = \forall x.\, x\geq \mathbf{0} \Rightarrow (\rel(x) \Leftrightarrow \rel(x+3))$	
\end{enumerate}

Notice that for $\textit{Axiom}_7$, it is not enough that a similar pattern appears within each interval $[\mathbf{3k}, \mathbf{3k+3})$: there must be an exact offset of $3$ with the previous interval. This is mandatory to connect two consecutive configurations and ensure that they are coherent with the transition relation of the Turing machine, as defined later. The same goes for $\textit{Axiom}_6$, where the exact offset of $1$ will allow to connect the position of the head to the tape content within a single configuration.

\begin{figure}[h]
	\centering
		\begin{tikzpicture}[scale=.7,node distance=2cm]
			\path[use as bounding box] (0, 4.7) rectangle (9, 5.3);
			
			\draw[dotted] (1.7,5) -- (1.95,5);
			\draw[dotted] (7.05,5) -- (7.35,5);
			
			\filldraw[] (0,4.8) rectangle ++(-1.5,0.4);
			\draw[] (0,4.8) rectangle ++(1.5,0.4);
			
			\filldraw[] (7.5,4.8) rectangle ++(1.5,0.4);
			\draw[] (9,4.8) rectangle ++(1.5,0.4);
			
			\filldraw[] (3.7,4.8) rectangle ++(0.8,0.4);
			\draw[] (4.5,4.8) rectangle ++(0.8,0.4);
			
			\filldraw[] (2.6,4.85) rectangle ++(0.4,0.3);
			\draw[] (3,4.85) rectangle ++(0.4,0.3);
			
			\filldraw[] (5.6,4.85) rectangle ++(0.4,0.3);
			\draw[] (6,4.85) rectangle ++(0.4,0.3);
			
			\filldraw[] (2,4.9) rectangle ++(0.2,0.2);
			\draw[] (2.2,4.9) rectangle ++(0.2,0.2);
			
			\filldraw[] (6.6,4.9) rectangle ++(0.2,0.2);
			\draw[] (6.8,4.9) rectangle ++(0.2,0.2);
			
			\draw[thick] (0,5.3) -- (0,4.7);
			\draw[thick] (9,5.3) -- (9,4.7);
			\node[] at (0,4.4) {$\mathbf{1}$};
			\node[] at (9,4.4) {$\mathbf{2}$};
			
			\draw[thick] (1.65,5.4) -- (1.65,4.6);
			\draw[thick] (7.4,5.4) -- (7.4,4.6);
			\node[] at (1.65,4.4) {$b_1$};
			\node[] at (7.4,4.4) {$b_2$};
		\end{tikzpicture}
		\caption{The points of the grid surrounded by open \textit{true} (black) and \textit{false} (white) intervals within $(\mathbf{1},\mathbf{2})$.}
		\label{fig:1predicate}
\end{figure}
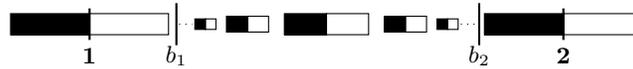

\noindent The formula $\paxiom = \underset{1 \leq k \leq 7}{\bigwedge} \textit{Axiom}_k$ axiomatizes the predicate $P$.

\begin{lemma}
	\label{lemma:consistency}
	The formula $\paxiom$ is consistent.
\end{lemma}

The proof sketch below provides the key ideas to construct a model of $\paxiom$. The complete construction is described in the Appendix.

\begin{proof}
Let us construct a subset $S$ of $\mathbb{R}$ that is a model of $\paxiom$. 
Firstly, we make every negative number belong to $S$, which ensures that there do not exist negative supporting points. The interval $[0,1]$ is then cut into $2N-2$ intervals of equal length, which alternate between being included in $S$, and being disjoint from $S$. This ensures the existence of exactly $N-1$ supporting points within the interval $(- \infty, 1)$, $0$ being the first; $1$ will be considered later. These $N-1$ supporting points are referred to as $s_1, s_2, \dots s_{N-1}$ and are depicted in Figure~\ref{fig:model01}. Recall that the supporting points are exactly those surrounded by an interval of $S$ (i.e., black on the figure) on the left, and an interval disjoint from $S$ (i.e., white) on the right.

\begin{figure}[t]
	\centering
	\begin{tikzpicture}[scale=.7,node distance=2cm]
		\path[use as bounding box] (-3.5, 4.6) rectangle (10, 5.4);
		
		\filldraw (0,4.8) rectangle ++(-3.5,0.4);
		
		\draw (0,4.8) rectangle ++(.75, 0.4);
		\filldraw (.75,4.8) rectangle ++(.75, 0.4);
		
		\draw (1.5,4.8) rectangle ++(.75, 0.4);
		\filldraw (2.25,4.8) rectangle ++(.75, 0.4);
		
		\draw (3,4.8) rectangle ++(.75, 0.4);
		\filldraw (3.75,4.8) rectangle ++(.75, 0.4);
		\draw (4.5,4.8) rectangle ++(.75, 0.4);
		
		\draw[thick, dotted] (5.25, 5) -- (7.75, 5);
		
		\filldraw (10,4.8) rectangle ++(-.75, 0.4);
		\draw (9.25,4.8) rectangle ++(-.75, 0.4);			
		\filldraw (8.5,4.8) rectangle ++(-.75, 0.4);
		
		\draw[thick] (1.5,5.3) -- (1.5,4.7);
		\draw[thick] (3,5.3) -- (3,4.7);
		\draw[thick] (4.5,5.3) -- (4.5,4.7);
		\draw[thick] (8.5,5.3) -- (8.5,4.7);
		\node[] at (0,5.6) {$s_1$};
		\node[] at (1.5,5.6) {$s_2$};
		\node[] at (3,5.6) {$s_3$};
		\node[] at (4.5,5.6) {$s_4$};
		\node[] at (8.5,5.6) {$s_{N-1}$};
		\node[] at (10,5.6) {$s_N$};
		
		\draw[thick] (0,5.3) -- (0,4.7);
		\draw[thick] (10,5.3) -- (10,4.7);
		\node[] at (0,4.4) {$0$};
		\node[] at (10,4.4) {$1$};
		\node[] at (-3.5,4.4) {$-\infty$};	
	\end{tikzpicture}
	\caption{A model for the axiomatization of $P$ over the interval $(-\infty, 1)$.}
	\label{fig:model01}
\end{figure}
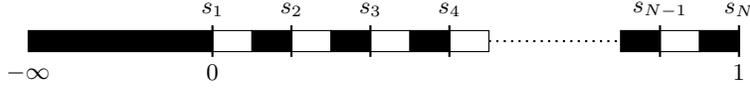

In order to make the real value $1$ the $N$-th supporting point, it is enough to make an interval on its right disjoint from $S$, e.g., the interval $(1,1+\frac{1}{4})$. Symmetrically, we make the interval $(2 - \frac{1}{4}, 2)$ included in $S$, satisfying the left part of the requirement for the real value $2$ to be a supporting point. \\
We further characterize $S$ such that the set of supporting points within the interval $(1+\frac{1}{4}, 2-\frac{1}{4})$ is $<$-isomorphic to $\mathbb{Z}$. This can be done by partitioning the open interval $(1+\frac{1}{4}, 2-\frac{1}{4})$ into a bi-infinite sequence of open intervals alternating between being included and disjoint from $S$, as depicted in Figure~\ref{fig:modelAxiomSupport}.

\begin{figure}[h]
	\centering
	\begin{tikzpicture}[scale=.9,node distance=2cm]
		\path[use as bounding box] (0, 4.6) rectangle (10, 6.5);
		\draw[] (0,4.75) rectangle ++(1.5,0.5);
		
		\filldraw[] (10,4.75) rectangle ++(-1.5,0.5);
		
		\draw[thick] (0,5.3) -- (0,4.7);
		\draw[thick] (10,5.3) -- (10,4.7);
		\node[] at (0,4.4) {$1$};
		\node[] at (10,4.4) {$2$};
		
		\filldraw[] (5,4.75) rectangle ++(-1.55,0.5);
		\draw[] (5,4.75) rectangle ++(1.55,0.5);
		
		\filldraw[] (3.05,4.8) rectangle ++(-.4,0.4);
		\draw[] (3.05,4.8) rectangle ++(.4,0.4);
		
		\filldraw[] (6.95,4.8) rectangle ++(-.4,0.4);
		\draw[] (6.95,4.8) rectangle ++(.4,0.4);
		
		\filldraw[] (2.45,4.85) rectangle ++(-.2,0.3);
		\draw[] (2.45,4.85) rectangle ++(.2,0.3);
		
		\filldraw[] (7.55,4.85) rectangle ++(-.2,0.3);
		\draw[] (7.55,4.85) rectangle ++(.2,0.3);
		
		\filldraw[] (2.15,4.9) rectangle ++(-.1,0.2);
		\draw[] (2.15,4.9) rectangle ++(.1,0.2);
		
		\filldraw[] (7.85,4.9) rectangle ++(-.1,0.2);
		\draw[] (7.85,4.9) rectangle ++(.1,0.2);
		
		\draw[dotted] (1.5, 5) -- (2.05, 5);
		\draw[dotted] (7.95, 5) -- (8.5, 5);
		
		\draw[thick] (1.5,5.4) -- (1.5,4.6);
		\draw[thick] (8.5,5.4) -- (8.5,4.6);
		\node[] at (1.5,4.3) {$1.25$};
		\node[] at (8.5,4.3) {$1.75$};
		\node[] at (1.5,5.7) {$b_1$};
		\node[] at (8.5,5.7) {$b_2$};
		\draw[thick] (5,5.4) -- (5,4.6);
		\node[] at (5,4.3) {$1.5$};
		
		\draw[] [decorate, decoration = {calligraphic brace}] (3.45,5.5) -- (5,5.5);
		\node[] at (4.225,5.9) {$\frac{1}{8}$};
		
		\draw[] [decorate, decoration = {calligraphic brace}] (2.65,5.5) -- (3.45,5.5);
		\node[] at (3.05,5.9) {$\frac{1}{16}$};
		
		\draw[] [decorate, decoration = {calligraphic brace}] (2.25,5.5) -- (2.65,5.5);
		\node[] at (2.45,5.9) {$\frac{1}{32}$};
		
		\draw[] [decorate, decoration = {calligraphic brace}] (1.5,6.2) -- (5,6.2);
		\node[] at (3.25,6.55) {$\frac{1}{4}$};
	\end{tikzpicture}
	\caption{A model for the axiomatization of $P$ over the interval $(1,2)$.}
	\label{fig:modelAxiomSupport}
\end{figure}
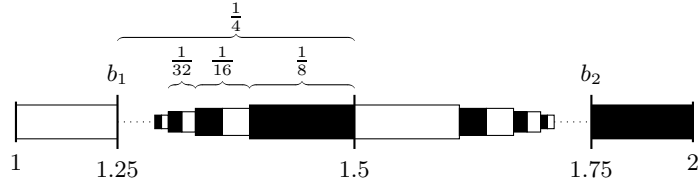
The whole pattern described on the interval $(1,2)$ can be directly transposed onto the interval $(2,3)$ with an exact offset of $+1$. Similarly, the distribution of $S$ over the interval $(0,3)$ can be transposed onto every interval $(3k,3k+3)$ with an offset of $+3k$, for $k>0$.
The only real values for which we do not describe their relation with $S$ are the points surrounded by an interval included in $S$ on one side, and an interval disjoint from $S$ on the other side. These points never conflict with the axiomatization $\paxiom$ which only deals with non-empty open intervals.

By construction, $S$ satisfies each axiom of the formula $\paxiom$, and is therefore a model of this formula. \qed
\end{proof}

\subsection{Encoding a Configuration of the Turing Machine}

Now that the supporting grid has been properly defined, the actual encoding of a given configuration can be addressed. That is, the state, the tape content and the head position of the $(k+1)$-th configuration of a run are encoded on the supporting points contained within the interval $[3k, 3k+3)$.

\subsubsection{Encoding the State}

Encoding the state of a given configuration is rather direct since we defined the grid to contain $N$ consecutive supporting points within every interval $[3k, 3k+1]$ for $k \in \mathbb{N}$, that can support the encoding of a state.
We only need to indicate that we start reading the encoding on a multiple of $3$. However the logic \rdluup does not allow to express periodicity constraints on variables.
Nevertheless, thanks to our axiomatization, $0$ and every other positive multiple of $3$ are the only points that simultaneously have no supporting predecessor, while admitting a supporting successor. These properties are expressible as follows: \\
$\mbox{\hspace{8pt}} \leftcvg(x) = \forall z.\, (z<x \land \rel(z)) \Rightarrow \exists y.\, z<y<x \land \rel(y)$ \\
$\mbox{\hspace{8pt}} \rsucc(x) = \exists z.\, x<z \land \rel(z) \land \forall y.\, x<y<z \Rightarrow \lnot \rel(y)$\\
For convenience, we introduce the formula \encodbeg to characterize a real value $x$ on which the encoding of a state starts: \\
$\mbox{\hspace{8pt}} \encodbeg(x) = \rel(x) \land \leftcvg(x) \land \rsucc(x)$ \\
Furthermore, the formula $\state{q}$ expresses that a state $q \in Q$ is encoded on a given real number $x$ and its $N-1$ supporting successors:
\begin{align*}
	\state{q}(x) = &\ \encodbeg(x) \land \exists y_1, \dots y_{N}.\, x = y_1 \\
 &  \quad \land  \bigwedge_{1\leq i<N} \textit{Succ}_{\textit{Sup}}(y_{i+1}, y_i) \land \bigwedge_{1\leq i \leq N} P(y_i) = b^q_i
\end{align*}
where $P(y_i) = b^q_i$ is a shorthand for $P(y_i)$ if $b^q_i = \top$, and $\lnot P(y_i)$ if $b^q_i = \bot$.

\subsubsection{Encoding the Head Position}
The position of the head is encoded in the second part of the grid, that is, in
the interval $(3k+1, 3k+2)$ for the $(k+1)$-th configuration
(cf. Figure~\ref{fig:DepictLinOrder}).  The grid on this interval is
$<$-isomorphic to $\mathbb{Z}$.  Each element of this subordering will
correspond to a position of the tape. When the predicate $P$ is true at
such a point, it means that the head points towards that cell. Since the Turing
machines that we consider here have a single read/write head, it must point
towards a unique cell for each configuration. Therefore $P$ must be true only
for a single element of that subordering.

\subsubsection{Encoding the Tape Content}
Similarly, the tape content is encoded in the third part of the grid, that is,
in the interval $(3k+2, 3k+3)$ for the $(k+1)$-th configuration
(cf. Figure~\ref{fig:DepictLinOrder}).  Again, the grid on this interval is
$<$-isomorphic to $\mathbb{Z}$. And again, each element $x$ of this subordering
will correspond to a cell of the tape, matching the cell that corresponds to $x-1$
in the head position interval. Figure~\ref{fig:schemeSuccConfig} illustrates the
connections between the suborderings, within a single configuration and with the
next one. The idea of the encoding is to simply set the value of $P$ to true
on the elements of the subordering that correspond to cells containing a $1$,
and to false for cells containing a $0$.

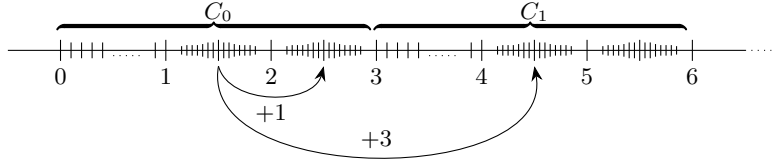
\begin{figure}[h]
	\centering
		\begin{tikzpicture}[scale=.7]
			\path[use as bounding box] (-1, -1.6) rectangle (13.5, .8);
			\draw (-1,0) -- (13,0);
			\draw[dotted] (13,0) -- (13.5,0);
			
			\draw (0,-.2) -- (0,.2);
			\draw (2,-.2) -- (2,.2);
			\draw (4,-.2) -- (4,.2);
			\draw (6,-.2) -- (6,.2);
			\draw (8,-.2) -- (8,.2);
			\draw (10,-.2) -- (10,.2);
			\draw (12,-.2) -- (12,.2);
			
			\node at (0,-.5) {$0$};
			\node at (2,-.5) {$1$};
			\node at (4,-.5) {$2$};
			\node at (6,-.5) {$3$};
			\node at (8,-.5) {$4$};
			\node at (10,-.5) {$5$};
			\node at (12,-.5) {$6$};
			
			\draw[ultra thick] [decorate, decoration = {calligraphic brace}] (-0.06,.4) -- (5.88,.4);
			\node[] at (3,.73) {$C_0$};
			\draw[ultra thick] [decorate, decoration = {calligraphic brace}] (5.96,.4) -- (11.88,.4);
			\node[] at (9,.73) {$C_1$};
			
			\draw (0.2,0.15) -- (0.2,-0.15);
			\draw (0.4,0.15) -- (0.4,-0.15);
			\draw (0.6,0.15) -- (0.6,-0.15);
			\draw (0.8,0.15) -- (0.8,-0.15);
			\draw[dotted] (1, -.1) -- (1.6, -.1);
			\draw (1.8,0.15) -- (1.8,-0.15);
			
			\draw (6.2,0.15) -- (6.2,-0.15);
			\draw (6.4,0.15) -- (6.4,-0.15);
			\draw (6.6,0.15) -- (6.6,-0.15);
			\draw (6.8,0.15) -- (6.8,-0.15);
			\draw[dotted] (7, -.1) -- (7.6, -.1);
			\draw (7.8,0.15) -- (7.8,-0.15);
			
			\draw (3,0.2) -- (3,-0.2);
			\draw (3.1,0.15) -- (3.1,-0.15);
			\draw (2.9,0.15) -- (2.9,-0.15);
			\draw (3.2,0.15) -- (3.2,-0.15);
			\draw (2.8,0.15) -- (2.8,-0.15);
			\draw (3.3,0.12) -- (3.3,-0.12);
			\draw (2.7,0.12) -- (2.7,-0.12);
			\draw (3.4,0.1) -- (3.4,-0.1);
			\draw (2.6,0.1) -- (2.6,-0.1);
			\draw (3.5,0.1) -- (3.5,-0.1);
			\draw (2.5,0.1) -- (2.5,-0.1);
			\draw (3.6,0.1) -- (3.6,-0.1);
			\draw (2.4,0.1) -- (2.4,-0.1);
			\draw (3.7,0.08) -- (3.7,-0.08);
			\draw (2.3,0.08) -- (2.3,-0.08);
			
			\draw (5,0.2) -- (5,-0.2);
			\draw (5.1,0.15) -- (5.1,-0.15);
			\draw (4.9,0.15) -- (4.9,-0.15);
			\draw (5.2,0.15) -- (5.2,-0.15);
			\draw (4.8,0.15) -- (4.8,-0.15);
			\draw (5.3,0.12) -- (5.3,-0.12);
			\draw (4.7,0.12) -- (4.7,-0.12);
			\draw (5.4,0.1) -- (5.4,-0.1);
			\draw (4.6,0.1) -- (4.6,-0.1);
			\draw (5.5,0.1) -- (5.5,-0.1);
			\draw (4.5,0.1) -- (4.5,-0.1);
			\draw (5.6,0.1) -- (5.6,-0.1);
			\draw (4.4,0.1) -- (4.4,-0.1);
			\draw (5.7,0.08) -- (5.7,-0.08);
			\draw (4.3,0.08) -- (4.3,-0.08);
			
			\draw (9,0.2) -- (9,-0.2);
			\draw (9.1,0.15) -- (9.1,-0.15);
			\draw (8.9,0.15) -- (8.9,-0.15);
			\draw (9.2,0.15) -- (9.2,-0.15);
			\draw (8.8,0.15) -- (8.8,-0.15);
			\draw (9.3,0.12) -- (9.3,-0.12);
			\draw (8.7,0.12) -- (8.7,-0.12);
			\draw (9.4,0.1) -- (9.4,-0.1);
			\draw (8.6,0.1) -- (8.6,-0.1);
			\draw (9.5,0.1) -- (9.5,-0.1);
			\draw (8.5,0.1) -- (8.5,-0.1);
			\draw (9.6,0.1) -- (9.6,-0.1);
			\draw (8.4,0.1) -- (8.4,-0.1);
			\draw (9.7,0.08) -- (9.7,-0.08);
			\draw (8.3,0.08) -- (8.3,-0.08);
			
			\draw (11,0.2) -- (11,-0.2);
			\draw (11.1,0.15) -- (11.1,-0.15);
			\draw (10.9,0.15) -- (10.9,-0.15);
			\draw (11.2,0.15) -- (11.2,-0.15);
			\draw (10.8,0.15) -- (10.8,-0.15);
			\draw (11.3,0.12) -- (11.3,-0.12);
			\draw (10.7,0.12) -- (10.7,-0.12);
			\draw (11.4,0.1) -- (11.4,-0.1);
			\draw (10.6,0.1) -- (10.6,-0.1);
			\draw (11.5,0.1) -- (11.5,-0.1);
			\draw (10.5,0.1) -- (10.5,-0.1);
			\draw (11.6,0.1) -- (11.6,-0.1);
			\draw (10.4,0.1) -- (10.4,-0.1);
			\draw (11.7,0.08) -- (11.7,-0.08);
			\draw (10.3,0.08) -- (10.3,-0.08);
			
			\draw[-{Stealth[length=2mm]}] (3,-.3) to [bend right=75] node[below] {$+1$} (5,-.3);
			\draw[-{Stealth[length=2mm]}] (3,-.3) to [bend right=100] node[above] {$+3$} (9,-.3);
		\end{tikzpicture}
	\caption{The first two consecutive configuration encodings.}
	\label{fig:schemeSuccConfig}
\end{figure}

\subsection{Enforcing a Valid Run}
Let us now define formally the formulas characterizing an accepting run of $\mathcal{M}$. We will decompose the global formula into three main parts: the initial conditions $\startformula$, the conditions on the transitions $\stepformula$ and the halting condition $\terminationformula$. For the sake of clarity, we use capital letters for these higher-level formulas. \\
The initial conditions of $\mathcal{M}$ are that the state encoded on $\mathbf{0}$ and its $N-1$ supporting successors is the initial state $q_0$, that the head points towards a unique initial unspecified cell of the tape, and finally that the tape is initially filled with $0$'s.
These conditions are expressed by the following formula:
\begin{align*}
	\startformula = \state{q_0}(\mathbf{0})\ \land\ & \big[\exists y.\, \mathbf{1}<y<\mathbf{2} \land \rel(y) \land P(y) \\
	& \quad \land \forall x.\, (\mathbf{1}<x<\mathbf{2} \land \rel(x) \land P(x)) \Rightarrow x=y\big] \\
	\land\ & \big[\forall y.\, (\mathbf{2}<y<\mathbf{3} \land \rel(y)) \Rightarrow \lnot P(y)\big]
\end{align*}

 The requirements on the transition are more complex. 
Intuitively, if before reaching the step $i \in \mathbb{N}$, we have not yet encountered the halting state $q_F$, then we must ensure that the configuration at Step $i$ can be obtained from the configuration at the previous step $i-1$ by following a transition \mbox{$(q, \alpha, q', \alpha', \lambda) \in \Delta$}.
The overall formula for this condition is the following:
\begin{align*}
	\stepformula = \forall y.\, (& y>0 \land \encodbeg(y) \land \textit{NotEnded}_\mathcal{M}(y)) \\
	& \Rightarrow \exists x.\, y = x + 3 \land \textit{Transition}_\mathcal{M}(x, y)
\end{align*}

The subformula $\textit{NotEnded}_\mathcal{M}(y)$ expresses that no valid real value prior to $y$ (i.e., a positive multiple of $3$ strictly smaller than $y$) encodes the halting state. This formula is defined by:
\begin{align*}
	\textit{NotEnded}_\mathcal{M}(y) = \forall x.\, (x<y \land \encodbeg(x)) \Rightarrow \lnot (\textit{State}_{q_F}(x))
\end{align*}

The subformula $\textit{Transition}_{\mathcal{M}}(x,y)$ expresses that there exists a transition $(q, \alpha, q', \alpha', \lambda) \in \Delta$ that allows to move in one step from the configuration encoded at $x$ (i.e., that the encoding of the configuration starts exactly on $x$), to the configuration corresponding to $y$. 
To improve readability, we decompose the condition on the transition relation as follows:
\begin{align*}
	&\textit{Transition}_\mathcal{M}(x, y) =\\
	&\qquad \bigvee_{(q, \alpha, q', \alpha', \lambda) \in \Delta}  \Big[  \state{q}(x) \land \state{q'}(y) \land \textit{Tape}_{\alpha, \alpha'}(x, y) \land \textit{Head}_{\lambda}(x, y) \Big]   
\end{align*}

For a given transition $(q, \alpha, q', \alpha', \lambda) \in \Delta$, the conditions on the states, tape and head are expressed as follows:

\begin{itemize}
	\item
	The state $q$ must be encoded on the real value $x$, and the state $q'$ on $y$: $\state{q}(x) \land \state{q'}(y)$ 

	\item
	The tape must contain $\alpha \in \{0,1\}$ at the position of the head for the
	step corresponding to $x$. Additionally, for the step corresponding to $y$, the tape must contain $\alpha'$ at
	the previous position of the head, and remain unchanged at all other positions.
	\begin{align*}
		&\textit{Tape}_{\alpha, \alpha'}(x, y) = \big[\forall z.\, (x+1<z<x+2 \land \rel(z) \land P(z)) \\
		&\mbox{\hspace{180pt}} \Rightarrow P(z+1) = \alpha \land P(z+4) = \alpha'\big] \\
		&\land \! \big[\forall z.\, (x+1<z<x+2 \land \rel(z) \land \lnot P(z)) \Rightarrow (P(z+1) \Leftrightarrow P(z+4))\big]
	\end{align*}

	Where $P(z+k) = \alpha$ is a shorthand for $\exists u.\, u = z+k \land P(u)$ if $\alpha = 1$, and $\exists u.\, u = z+k \land \lnot P(u)$ if $\alpha = 0$. The $``+1"$ operator allows us to connect the encoding of the head position with the encoding of the tape content within the same configuration. The $``+4"$ operator does the same while jumping to the next configuration (cf. Figure 4).
	Notice that this formula does not involve $y$; it assumes (rightfully, given the formula $\stepformula$) that the equality $y = x+3$ holds.
	
	\item
	The head is moved in the direction specified by $\lambda \in \{L, R\}$, i.e., left for $L$ and right for $R$. This can be expressed by exploiting the predecessor and successor relations defined for supporting real values.
	\begin{align*}
		\textit{Head}_{\lambda}(x, y) = \, \forall z.\, &(x+1<z<x+2 \land \rel(z) \land P(z)) \\
		& \quad \Rightarrow \exists v.\, f_\lambda(v, z+3) \land P(v) \land \lnot P(z)
	\end{align*}
	where $f_R = \onesucc$ and $f_L = \onepred$.
	Since in the initial configuration of the Turing machine the head points towards a single cell, the formula $\textit{Head}_{\lambda}$ ensures that this remains the case throughout every run of the Turing machine.	
\end{itemize}

Finally, the existence of a halting run is expressed by the formula:
\begin{align*}
	\textit{END}_{\mathcal{M}} = \exists x.\, \state{q_F}(x)
\end{align*}

The global formula that expresses that the Turing machine $\mathcal{M}$
halts on some run encoded by the value of the predicate $P$ is the following:
\begin{align*}
	\textit{HALT}_{\mathcal{M}} = \textit{START}_{\mathcal{M}} \land \textit{STEP}_{\mathcal{M}} \land \textit{END}_{\mathcal{M}} \land \paxiom
\end{align*}

\noindent where $\paxiom$ is the axiomatization of the supporting points as described in Section~\ref{subsec:support}. 

By construction, satisfiability of the global formula $\haltformula$ is equivalent to the existence of a halting run for the Turing machine $\mathcal{M}$. It follows that the satisfiability problem for \rdluup is undecidable, which proves Theorem~\ref{thm:undec}.

\section{Conclusion}
\label{sec:conclu}

This work provides a lower and an upper bound for the decidability of first-order
fragments with quantifiers mixing uninterpreted unary predicates and weak forms
of real arithmetic.  This draws a precise picture of the
frontier of decidability in fragments mixing real arithmetic and uninterpreted
predicates.

We proved the decidability of the fragment \lmix, where uninterpreted unary
predicates, order constraints between real and integer variables, and
difference logic constraints between integer variables are allowed.  This result
is a consequence of the already established decidability of its restriction
\msoor, where only uninterpreted unary predicates and order constraints between
real values are allowed. To the best of our knowledge, there does not exist yet 
a practical decision procedure for \msoor.
It would be interesting to investigate which arithmetic extensions of this decidable fragment preserve decidability. Note however that our proof of decidability relies on the translation of the fragment constraints into the first-order theory of order over $\mathbb{R}$, with unary predicates, which is not directly feasible for, e.g., constraints of the form $x+y \bowtie 0$, where $x$ and $y$ are variables, and $\bowtie \ \in \{<, \leq, =, \geq, >\}$.

In another result, we established the undecidability of the fragment \rdluup, where
uninterpreted unary predicates and difference logic constraints between real
variables are allowed. It is worth mentioning that this result can be adapted straightforwardly to the same logic interpreted over the domain $\mathbb{Q}$.

Our long term goal is to design an effective decision procedure for the decidable fragment.  Complexity results have been established~\cite{rabinovich2012temporal,cristau2009automata,reynolds2010complexity} for the temporal logic counterpart of the theory of order, to which we reduce the decidability of our fragment of interest. We are currently designing a decision procedure relying on the concept of automata on linear orderings introduced in~\cite{bruyere2007automata}.
We hope that the insight we obtained through this decision procedure will eventually guide the design of new powerful instantiation techniques for SMT in a more expressive context, and that these techniques will happen to be complete in particular for this decidable fragment.

\vspace*{5pt}\noindent
\textit{Acknowledgments:}  We are thankful to Tanja Schindler and the reviewers of 
this paper and of our previous work-in-progress workshop paper for their comments.


\bibliographystyle{splncs04}
\bibliography{refs}

\clearpage\section*{Appendix}
\label{app:A}

In order to formally prove Lemma~\ref{lemma:consistency}, we provide a complete description of a model of the formula $\paxiom$. 

\begin{proof}
	In order to prove that the axiomatization is consistent, let us construct a model for $\paxiom$.
	Let $S$ be a subset of $\mathbb{R}$ such that:
	\begin{itemize}[$\bullet$]
		\item All strictly negative numbers belong to $S$:	$(-\infty, 0) \subseteq S$.
		
		\item The interval $[0,1]$ is cut into $2N-2$ intervals of equal length, that are either included or disjoint from $S$:
		\begin{itemize}[-]	
			\item $(\frac{i}{N-1}, \frac{2i+1}{2(N-1)}) \subseteq \mathbb{R} \setminus S$ for all $i \in [0,N-2]$,
			\item $(\frac{2i+1}{2(N-1)}, \frac{i+1}{N-1}) \subseteq S$ for all $i \in [0,N-2]$.
		\end{itemize}
		This distribution ensures the existence of exactly $N-1$ supporting points ($0$ is the first one, and $1$ will be considered later). These $N-1$ supporting points are referred to as $s_1, s_2, \dots s_{N-1}$, and they are depicted in Figure~\ref{fig:model01}. Recall that the supporting points are exactly those surrounded by an interval of $S$ (i.e., black on the figure) on the left, and an interval of the complement of $S$ (i.e., white) on the right.
		
		\item For the interval $[1,2]$, we first ensure that $1$ is also a supporting point by making the open interval $(1, 1+\frac{1}{4})$ on its right belong to $\mathbb{R} \setminus S$. We also make the interval $(2-\frac{1}{4},2)$ belong to $S$, to deal with the left part of the requirement of $2$ being a supporting point. 
		
		Furthermore, in order to generate the desired set of supporting points within $(1,2)$ (that must be \mbox{$<$-isomorphic} to $\mathbb{Z}$), we partition the interval  $(1 + \frac{1}{4}, 2 - \frac{1}{4})$ into a bi-infinite sequence of disjoint open subintervals that alternate between being included and disjoint from $S$, as depicted in Figure~\ref{fig:modelAxiomSupport}: 
		\begin{itemize}
			\item $(1 , 1 + \frac{1}{4}) \subseteq \mathbb{R} \setminus S$, and $(2 - \frac{1}{4}, 2) \subseteq S$,
			\item $(1 + \frac{3}{8}, 1 + \frac{1}{2}) \subseteq S$ and $(1 + \frac{1}{2}, 1 + \frac{5}{8}) \subseteq \mathbb{R} \setminus S$,
			\item $(1 + \frac{1}{2} - \sum\limits_{i=3}^{j}\frac{1}{2^i}, 1 + \frac{1}{2} - \sum\limits_{i=3}^{j}\frac{1}{2^i} + \frac{1}{2^{j+2}}) \subseteq S$ for all $j \geq 4$,
			\item $(1 + \frac{1}{2} - \sum\limits_{i=3}^{j}\frac{1}{2^i} + \frac{1}{2^{j+2}}, 1 + \frac{1}{2} - \sum\limits_{i=3}^{j-1}\frac{1}{2^i}) \subseteq \mathbb{R} \setminus S$ for all $j \geq 4$,
			\item $(1 + \frac{1}{2} + \sum\limits_{i=3}^{j}\frac{1}{2^i}, 1 + \frac{1}{2} + \sum\limits_{i=3}^{j}\frac{1}{2^i} + \frac{1}{2^{j+2}}) \subseteq S$ for all $j \geq 3$,
			\item $(1 + \frac{1}{2} + \sum\limits_{i=3}^{j}\frac{1}{2^i} + \frac{1}{2^{j+2}}, 1 + \frac{1}{2} + \sum\limits_{i=3}^{j+1}\frac{1}{2^i}) \subseteq \mathbb{R} \setminus S$ for all $j \geq 3$.	
		\end{itemize}
		Notice that the union of all the intervals for which we state their relation to $S$, covers the interval $(1,2)$ entirely apart from isolated points that are surrounded by an interval included in $S$, and an interval disjoint from $S$. Some of them are supporting points (when the interval included in $S$ is on the left, and the interval disjoint from $S$ is on the right), the others are completely irrelevant.
		
		\item The pattern described on the interval $(1,2)$ can be directly transposed onto the interval $(2,3)$, with an exact offset of $+1$:
		\begin{itemize}
			\item $(2 , 2 + \frac{1}{4}) \subseteq \mathbb{R} \setminus S$, and $(3 - \frac{1}{4}, 3) \subseteq S$,
			\item $(2 + \frac{3}{8}, 2 + \frac{1}{2}) \subseteq S$ and $(2 + \frac{1}{2}, 2 + \frac{5}{8}) \subseteq \mathbb{R} \setminus S$,
			\item $(2 + \frac{1}{2} - \sum\limits_{i=3}^{j}\frac{1}{2^i}, 2 + \frac{1}{2} - \sum\limits_{i=3}^{j}\frac{1}{2^i} + \frac{1}{2^{j+2}}) \subseteq S$ for all $j \geq 4$,
			\item $(2 + \frac{1}{2} - \sum\limits_{i=3}^{j}\frac{1}{2^i} + \frac{1}{2^{j+2}}, 2 + \frac{1}{2} - \sum\limits_{i=3}^{j-1}\frac{1}{2^i}) \subseteq \mathbb{R} \setminus S$ for all $j \geq 4$,
			\item $(2 + \frac{1}{2} + \sum\limits_{i=3}^{j}\frac{1}{2^i}, 2 + \frac{1}{2} + \sum\limits_{i=3}^{j}\frac{1}{2^i} + \frac{1}{2^{j+2}}) \subseteq S$ for all $j \geq 3$,
			\item $(2 + \frac{1}{2} + \sum\limits_{i=3}^{j}\frac{1}{2^i} + \frac{1}{2^{j+2}}, 2 + \frac{1}{2} + \sum\limits_{i=3}^{j+1}\frac{1}{2^i}) \subseteq \mathbb{R} \setminus S$ for all $j \geq 3$.	
		\end{itemize}
		
		\item Finally, the distribution of $S$ over the interval $(0,3)$ defined above can be transposed to every interval $(3k,3k+3)$ for $k > 0$:
		\begin{itemize}
			\item $(3k+\frac{i}{N-1}, 3k+\frac{2i+1}{2(N-1)}) \subseteq \mathbb{R} \setminus S$ for all $i \in [0,N-2]$,
			\item $(3k+\frac{2i+1}{2(N-1)}, 3k+\frac{i+1}{N-1}) \subseteq S$ for all $i \in [0,N-2]$.
			\item $(3k + 1 , 3k + \frac{1}{4}) \subseteq \mathbb{R} \setminus S$, and $(3k + 2 - \frac{1}{4}, 3k + 2) \subseteq S$,
			\item $(3k + 1 + \frac{3}{8}, 3k + 1 + \frac{1}{2}) \subseteq S$ and $(3k + 1 + \frac{1}{2}, 3k + 1 + \frac{5}{8}) \subseteq \mathbb{R} \setminus S$,
			\item $(3k + 1 + \frac{1}{2} - \sum\limits_{i=3}^{j}\frac{1}{2^i}, 3k + 1 + \frac{1}{2} - \sum\limits_{i=3}^{j}\frac{1}{2^i} + \frac{1}{2^{j+2}}) \subseteq S$ for all $j \geq 4$,
			\item $(3k + 1 + \frac{1}{2} - \sum\limits_{i=3}^{j}\frac{1}{2^i} + \frac{1}{2^{j+2}}, 3k + 1 + \frac{1}{2} - \sum\limits_{i=3}^{j-1}\frac{1}{2^i}) \subseteq \mathbb{R} \setminus S$ for all $j \geq 4$,
			\item $(3k + 1 + \frac{1}{2} + \sum\limits_{i=3}^{j}\frac{1}{2^i}, 3k + 1 + \frac{1}{2} + \sum\limits_{i=3}^{j}\frac{1}{2^i} + \frac{1}{2^{j+2}}) \subseteq S$ for all $j \geq 3$,
			\item $(3k + 1 + \frac{1}{2} + \sum\limits_{i=3}^{j}\frac{1}{2^i} + \frac{1}{2^{j+2}}, 3k + 1 + \frac{1}{2} + \sum\limits_{i=3}^{j+1}\frac{1}{2^i}) \subseteq \mathbb{R} \setminus S$ for all $j \geq 3$,
			\item $(3k + 2 , 3k + 2 + \frac{1}{4}) \subseteq \mathbb{R} \setminus S$, and $(3k+3  - \frac{1}{4}, 3k+3) \subseteq S$,
			\item $(3k + 2 + \frac{3}{8}, 3k + 2 + \frac{1}{2}) \subseteq S$ and $(3k + 2 + \frac{1}{2}, 3k + 2 + \frac{5}{8}) \subseteq \mathbb{R} \setminus S$,
			\item $(3k + 2 + \frac{1}{2} - \sum\limits_{i=3}^{j}\frac{1}{2^i}, 3k + 2 + \frac{1}{2} - \sum\limits_{i=3}^{j}\frac{1}{2^i} + \frac{1}{2^{j+2}}) \subseteq S$ for all $j \geq 4$,
			\item $(3k + 2 + \frac{1}{2} - \sum\limits_{i=3}^{j}\frac{1}{2^i} + \frac{1}{2^{j+2}}, 3k + 2 + \frac{1}{2} - \sum\limits_{i=3}^{j-1}\frac{1}{2^i}) \subseteq \mathbb{R} \setminus S$ for all $j \geq 4$,
			\item $(3k + 2 + \frac{1}{2} + \sum\limits_{i=3}^{j}\frac{1}{2^i}, 3k + 2 + \frac{1}{2} + \sum\limits_{i=3}^{j}\frac{1}{2^i} + \frac{1}{2^{j+2}}) \subseteq S$ for all $j \geq 3$,
			\item $(3k + 2 + \frac{1}{2} + \sum\limits_{i=3}^{j}\frac{1}{2^i} + \frac{1}{2^{j+2}}, 3k + 2 + \frac{1}{2} + \sum\limits_{i=3}^{j+1}\frac{1}{2^i}) \subseteq \mathbb{R} \setminus S$ for all $j \geq 3$.	
		\end{itemize}
	\end{itemize}
	The only real values for which we do not describe their relation with $S$, are isolated points (some of them constituting the support, some not). They never conflict with the axiomatization $\paxiom$ which only deals with non-empty open intervals.
	
	By construction, $S$ satisfies each axiom of the formula $\paxiom$, and is therefore a model. \qed
	
\end{proof}

\end{document}